\documentclass[prl,twocolumn,superscriptaddress]{revtex4}
\usepackage{amsmath,amsthm,graphicx,multirow,bbold,hyperref}
\newcommand{\half}{\mbox{$\textstyle \frac{1}{2}$}}

\newcommand{\av}[1]{\langle #1 \rangle}
\newcommand{\bs}[1]{\boldsymbol{#1}}
\newcommand{\be}{\begin{equation}}
\newcommand{\ee}{\end{equation}}
\newcommand{\bea}{\begin{eqnarray}}
\newcommand{\eea}{\end{eqnarray}}
\begin{document}
\title{Limited measurement dependence in multiple runs of a Bell test}
\author{James E. Pope}
\affiliation{Mathematical Institute, University of Oxford, 24-29 St Giles', OX1 3LB, UK}
\author{Alastair Kay}
\affiliation{Department of Mathematics, Royal Holloway University of London, Egham, Surrey, TW20 0EX, UK}
\affiliation{Keble College, Parks Road, Oxford, OX1 3PG, UK}


\begin{abstract}
The assumption of free will -- the ability of an experimentalist to make random choices -- is central to proving the indeterminism of quantum resources, the primary tool in quantum cryptography. Relaxing the assumption in a Bell test allows violation of the usual classical threshold by correlating the random number generators used to select measurements with the devices that perform them. In this paper, we examine not only these correlations, but those across multiple runs of the experiment. This enables an explicit exposition of the optimal cheating strategy and how the correlations manifest themselves within this strategy. Similar to other recent results, we prove that there remain Bell violations for a sufficiently high, yet non-maximal degree of free will which cannot be simulated by a classical attack, regardless of how many runs of the experiment those choices are correlated over.
\end{abstract} \maketitle

\newtheorem{theorem}{Theorem}
\section{I. Introduction}
Bell's theorem \cite{1} provides an experimentally falsifiable prediction for certain correlations if Nature is deterministic. That these inequalities are found to be violated \cite{2,3} …constitutes proof of the incompatibility of classical, deterministic or stochastic, theories with the Universe, no matter that our knowledge of theories compatible with Nature may be incomplete. A definitive Bell test, free of loopholes, is yet to be realised. Nevertheless, the overwhelming consensus is that the correlations predicted by quantum theory have been verified since the common loopholes of locality \cite{1} and detection \cite{4} have been closed separately \cite{5,6}, and Nature would be strange indeed if it conspired to utilise whichever loophole were available in order to mask its classicality.

This violation of a Bell inequality as a proof technique has since been elevated to the central tool in proving the absence of an eavesdropper \cite{7}, being so powerful as to secure cryptographic schemes that need not be reliant on either the quantum theory that inspired them \cite{8} or complete knowledge of the devices used to implement them \cite{9}. This device-independent cryptography seeks security even when the users' devices are assumed to be controlled by an adversary, for tasks such as key distribution (see \cite{9} and references therein) or randomness expansion \cite{10}. However, this also elevates the stringent requirements of loophole closure; an adversary will certainly conspire to use every tool available to mask the classicality induced by their eavesdropping. This includes subverting not only the detectors and any locality weaknesses, but also corrupting any other tools the cryptographers might import into their laboratory, such as random number generators (RNGs). This corruption must be quite specific, since the choices of input to the Bell test made by the RNGs should still give the experimenters, ignorant of any adversarial involvement, the impression of perfect randomness.

Such corruption necessitates the study of `free will' loopholes \cite{11,12,13,14,15} in which the random numbers are not perfectly random, and an eavesdropper can use that knowledge to modify her strategy. The RNGs are characterised by an appropriate measure of the experimenters' free will in choosing their measurements, also known as measurement independence, though here we will use the term measurement dependence (MD) for reasons apparent in the definitions below. While there is no known way to experimentally determine this value, it remains important to understand, for a prescribed degree of MD, how much advantage can be gained by an adversary (or was gained previously), or how to exclude such influences. The latter question has recently been addressed in the form of randomness amplification protocols \cite{16,17,18} in which a random input string with a given MD is processed into a new random string about which any adversary has less information. These studies have either assumed no correlations between different runs of the experiment \cite{15,19} or \cite{16,17,18} restricted the probability distributions to be of a very specific form known as a Santha-Vazirani source \cite{20}, which requires for a source of bits $z_i \in \{0,1\}$ that for some $\epsilon>0$ and any $n$, \begin{equation*} \half - \epsilon \leq p(z_n|\lambda, z_1,\ldots,z_{n-1}) \leq \half + \epsilon \end{equation*} where $\lambda$ encapsulates a local variable influencing the source. Clearly, a positive lower bound on such probabilities prevents the predetermined exclusion of a measurement choice. Since this exclusion of a measurement choice is intimately involved with the optimal cheating strategy, other measures of MD could exhibit substantially different behaviour. Furthermore, whilst the users of the devices view the two RNGs as separate entities, and the Santha-Vazirani specification is for individual RNGs, an eavesdropper who is preprogramming this has arbitrary access to program them as she wishes, and so effectively considers them as one joint entity.

In this paper, we assign the experimenters a fixed degree of MD in making their measurement choices, using a measure not constrained by the Santha-Vazirani condition, and determine the maximum value that a classical strategy could possibly achieve in the Bell test, comparing that to the standard threshold of a Bell test. This has previously been examined with regards to attacking single runs of a CHSH test \cite{13}, more general Bell tests \cite{15}, and its application to randomness expansion \cite{19}. We show that an eavesdropper gains an advantage by correlating the partially random generation of measurement choices over many runs of the test. Our results explicitly describe the optimal correlations between RNGs and measurement devices that an adversary might introduce. We compare this strategy to the optimal quantum strategy, allowing us to prove that when there is a small (yet non-minimal) amount of MD, a sufficiently high Bell violation will exclude the possibility of a classical correlated attack. We outline our methods and analytic bounds for the bipartite Bell test due to Clauser, Horne, Shimony and Holt (CHSH) \cite{21}, and discuss other Bell tests which are equally amenable to the same numerical analysis.

\section{II. Measures of Measurement Dependence} The CHSH test consists of two parties, Alice and Bob, making random choices $j,k\in \{0,1\}$, corresponding to making one of two measurements, $A_j,\,B_k$ and obtaining outcomes, $a_j ,b_k\in \{\pm 1\}$. After recording the result of each measurement, they communicate in order to calculate
$$
\langle S\rangle=\langle a_0b_0\rangle+\langle a_0b_1\rangle+\langle a_1b_0\rangle-\langle a_1b_1\rangle.
$$
Assuming they have perfect RNGs, each measurement is equally likely, and the expected value of each measurement result (independent of what the other party measured) is 0. In the cryptographic scenario, Alice and Bob are trying to use this to prove that they have a quantum resource, the ideal result being $\langle S\rangle=2\sqrt{2}$, although no value $|\langle S\rangle|>2$ can be explained by somebody entirely replacing the quantum functionality of the box with a deterministic protocol. This changes, however, if the eavesdropper can manipulate the random measurement choices. Nevertheless, we will constrain the eavesdropper to strategies which, on average, cause each measurement choice to be equally likely, and each measurement outcome to be equally likely, otherwise Alice and Bob would soon spot that something was going on!

We are interested in assessing the performance of an eavesdropper when they are allowed to pre-program both the measurement devices and the RNGs of the two parties. Their collective strategy is a program, known as a local hidden variable (LHV) model, that the eavesdropper designs; a random variable $x\in \mathcal{X}$ specifying how to select the measurement bases and the corresponding results. The degree of control that the eavesdropper has over the choice of measurement basis is contained in the probabilities of making a given choice, $p(A_j,B_k|x)$. Numerous different ways have been proposed to assign a numerical value based on this \cite{13,14,15,19}. Perhaps the most natural class of measures are those that can be interpreted as the advantage that is gained by the eavesdropper's knowledge of the probability distribution as compared to that of Alice and Bob's:
$$
M_p=\max_{x\in \mathcal{X}}\left( \sum_{j,k}|p(A_j,B_k|x)-p(A_j,B_k)|^p\right)^{1/p},
$$
particularly for $p=1$. We require that the marginal distributions representing the measurement choices are uniform, $p(A_j,B_k)=\sum_x p(x)p(A_j,B_k|x)=\frac{1}{4}$, as is typically the case in a CHSH test, although there may be advantages to the cryptographers to lifting this expectation \cite{15}. Whilst the $p=1$ norm is amenable to analysis using linear programming \cite{22}, the choice of $p=\infty$ generalises to the multiple run case more readily and seems natural with its prominent ties to the min-entropy measure that is useful in cryptographic scenarios. This is seen by
\begin{equation*}
\lim_{p\rightarrow \infty} M_p= \max_{x,j,k} |p(A_j,B_k|x)-p(A_j,B_k)| = P - \frac{1}{4},
\end{equation*}
where $P$ is the maximum probability
\begin{equation*}
P:= \max_{j,k,x} p(A_j,B_k|x),
\end{equation*}
since the maximum is obtained for $p(A_j,B_k|x)>1/4$ (this will always be true in the large run limit in the parameter regime that is interesting for operation). This measure was introduced in \cite{19}, with $P=1$ representing the possibility of an entirely deterministic selection with no free will, and $P=\frac{1}{4}$ delivering the uniform measurement selections Alice and Bob expect to observe. The techniques presented here can also be applied to the measure introduced by Hall \cite{13}, subject to some minor technical adjustments. In contrast to the previous terminology, we say that these measures characterise measurement dependence (rather than independence), since a model with a higher evaluation displays measurement selections that are more dependent on the underlying variables.

\section{III. Optimal One-Shot Attack} We focus on maximising the score of a CHSH game subject to a fixed MD $P$. The measurement settings are $A_j ,B_k$ with $j,k\in \{0,1 \}$ and the outcomes are deterministically specified by underlying variables $x$, i.e.\ $a_j(x),b_k(x)\in\{\pm 1\}$, such that the game evaluates
\begin{equation} \label{CHSHscore}
S=4\sum_{x} p(x)\sum_{j,k\in \{0,1\}} p(A_j,B_k|x ) (-1)^{jk} a_j(x) b_k(x).
\end{equation}
This score is related to the probability of winning a single round of the CHSH game by $p_{\mathrm{win}}= (1+S/4)/2$. A single run of the experiment is said to give a `correct answer' to the CHSH game if either $(j,k)\neq (1,1)$ and the outcomes are equal, or $(j,k)=(1,1)$ and the outcomes are different. From a choice of 16 distinct outcome sets, half achieve the maximum local CHSH score, giving only one incorrect answer for the four possible query pairs. These are given by the four variables in Table \ref{CHSHoutcomestable}, along with their conjugates that specify the negative outcomes (these should be used half of the time to avoid suspicion of fixed outcomes, but do not affect correlations and have been suppressed from the calculations for simplicity).
\begin{table}[!tb]
\begin{center}
\begin{tabular}{c|cccc}
$x$ & $a_0$ & $a_1$ & $b_0$ & $b_1$ \\ \hline
$0$ & $1$ & $1$ & $1$ & $1$ \\
$1$ & $1$ & $-1$ & $1$ & $1$ \\
$2$ & $1$ & $1$ & $1$ & $-1$ \\
$3$ & $1$ & $-1$ & $-1$ & $1$ \\
\end{tabular}
\end{center}
\vspace{-0.5cm}
\caption{Outcomes specified by an underlying variable $x$ \label{CHSHoutcomestable}}
\vspace{-0.5cm}
\end{table}

The adversarial strategy selects each $x\in \{0,1,2,3\}$ with probability $1/4$, uniquely defining a set of predetermined outcomes, followed by the measurement choices $j$,$k$ to be used, represented by $y= 2j+k\in \{ 0,1,2,3\}$ so that the conditional probabilities may be rewritten as $p(y|x)$, and the definition of MD as $P=\max_{x,y} p(y|x)$. We seek to maximise the CHSH score \eqref{CHSHscore}
\begin{equation}
S=4-2\sum_{x+y=3} p(y|x) ,
\end{equation} subject to a fixed degree of MD $P$ and Bayes' theorem
\begin{equation}
\sum_x p(x)p(y|x)=p(y) \qquad \forall y ,\label{Bayescondition}
\end{equation}
which reduces to $\sum_x p(y|x) =1$ by the assertion that $p(x)=p(y)=1/4$. The optimal strategy has been derived using both Hall's measure \cite{13,23} and $P$ \cite{19}. The latter takes $P\in [1/4 ,1/3]$ such that $p(y|x)=P$ if $x+y\neq 3$ (i.e.\ correct answer) and $p(y|x)=1-3P$ otherwise, yielding $S=24P-4$ up to the threshold value $P=1/3$ that achieves maximum $S=4$. Optimality derives directly from the reasoning outlined below for the general case.

\section{IV. Multiple Runs}
\subsection{A. Classical Adversary}
Under the usual assumption of perfect measurement independence, the best classical eavesdropping strategy acts independently on each run of the experiment \cite{24}. A limited MD scenario necessitates re-investigation of the possible types of attack, since Eve may use her extra knowledge of the underlying system's imperfections to correlate her strategy appropriately. This is similar to LHV models which exploit imperfect photon detection rates by simulating a detection failure and changing the output strategy accordingly \cite{25}.

Assume, as before, that Eve has written a program that tells Alice's and Bob's devices what to do (using the local measurement settings as inputs), now encoding instructions for blocks of $N$ runs together. If Eve seeks to optimise the CHSH score, the outcomes for each run must still be drawn from Table \ref{CHSHoutcomestable}. The $N$-run system is fully characterised by the strings $\bs{x},\bs{y}\in \{0,1,2,3 \}^N$ such that $x_n$ and $y_n$ denote respectively the outcome assignments and pair of measurement settings for the $n$th run. The definition of $P$ intuitively extends to an $N$-run model by considering all possible combinations of measurement settings and underlying outcome specifications, \be P_{(N)}:= \max_{\bs{x},\bs{y}} p(\bs{y}|\bs{x}). \ee

It is clear that the optimal single-shot strategy above, when repeated independently for $N$ runs, has MD $P_{(N)}=P^N \in [4^{-N},3^{-N}]$. The optimal correlated attack will obviously perform as well or better than this, and comparison can be made with the repeated single-shot strategy over varying $N$ by taking the $N$th root, thus a fixed MD $P$ requires that $p(\bs{y}|\bs{x})\leq P^N$ for all $\bs{x},\bs{y}$.

The one-shot attack distinguishes which $(x_n ,y_n)$-pairs give a correct answer; the extension to $N$ runs asks \emph{how many} answers for a pair $(\bs{x},\bs{y})$ are correct, $k(\bs{x},\bs{y}):=\sum_{n=1}^N \delta_{x_n +y_n \neq 3}$. The average CHSH score is then 
\begin{equation}
S=-4+8\sum_{\bs{x}} p(\bs{x}) S_{\bs{x}} \qquad S_{\bs{x}} :=\frac{1}{N} \sum_{\bs{y}} p(\bs{y}|\bs{x})k(\bs{x},\bs{y}). \label{Sxydef}
\end{equation}
We wish to maximise $S$ subject to $p(\bs{x})=p(\bs{y})=4^{-N}$, Bayes' rule condition \be \sum_{\bs{x}} p(\bs{x})p(\bs{y}|\bs{x})=p(\bs{y}) \qquad \forall \bs{y} \label{BayesconditionN}\ee and the limited MD constraint $p(\bs{y}|\bs{x})\leq P^N$ for all $\bs{x},\bs{y}$. $S_{\bs{x}}$ may be rewritten as
\begin{equation*}
S_{\bs{x}}=\frac{1}{N}\sum_{k=0}^N kp_{k}^{\bs{x}} \qquad p_{k}^{\bs{x}}:=\sum_{\bs{y}:k(\bs{x},\bs{y})=k}p(\bs{y}|\bs{x}).
\end{equation*}
Since $p(\bs{y}|\bs{x})$ can be individually varied, the optimisation can be made on each $S_{\bs{x}}$ separately. The outcome specifications obtained from extending Table I to $N$ runs exhibit the following relation for any $k,\bs{x},\bs{x'}$:
\begin{equation*} \#\{\bs{y}:k(\bs{x},\bs{y})=k\}=\#\{\bs{y}:k(\bs{x'},\bs{y})=k\}=\binom{N}{k}3^k .\end{equation*} 
Thus, redistribution of the probabilities $p(\bs{y}|\bs{x})$ over any $\bs{y}$ for which $k(\bs{x},\bs{y})=k$ will have no effect on the maximisation. Optimisation of S corresponds to independent optimisation over each of the $S_x$. Since these quantities are the same for all $\bs{x},\bs{x'}$, it is evident that they have the same optimum, $p^{\bs{x}}_k=p^{\bs{x'}}_k$, so we may remove the $\bs{x}$ dependence, defining $p_k$ such that $p_{k}^{\bs{x}}=p_k \binom{N}{k}3^k$, thus \eqref{Sxydef} becomes
\be \label{Sinpkterms} S=\frac{8}{N} \sum_{k=0}^N k p_k 3^k \binom{N}{k} - 4 .\ee
The probabilities are also subject to Bayes' rule, condition \eqref{BayesconditionN} which, by the assertion that $p(\bs{x})=p(\bs{y})=4^{-N}$, reduces to
\be \label{norminpkterms} \sum_{k=0}^N p_k 3^k \binom{N}{k} = 1 ,\ee
whilst fixed MD $P$ requires that $p_k \leq P^N$ for all $k$. This problem can be solved by linear programming, which confirms the following argument.

Intuitively for a CHSH test, obtaining the maximum $S$ value requires more weight to be given to the pairs of measurement settings that answer a larger proportion of the $N$ queries correctly. Therefore we assign $p_N = P^N$. If the normalisation \eqref{norminpkterms} allows, set $p_{N-1}=P^N$, and so on. All remaining $p_k$ are set to 0. The curve of maximum $S$ against $P^N$ is piecewise linear, connected by the $N+1$ points defined by a parameter $l'$ such that
\bea \label{Pinlterms} P &=& \left( \sum_{k=l'}^N 3^k \binom{N}{k} \right)^{-1/N}, \\ \label{Sinlterms} S &=& \frac{8P^N}{N} \sum_{k=l'}^N k3^k \binom{N}{k}-4 .
\eea
The curve is linearly interpolated for $P^N$, where $P\in [1/4,1/3]$, between these points by assigning $p_k=0$ for $k<l'-1$, $p_k=P^N$ for $k\geq l'$, and letting $p_{l'-1}$ fulfil the normalisation \eqref{norminpkterms}. Figure \ref{fig:main} shows such plots for various finite values of $N$. We see immediately that the eavesdropper gains an advantage with increasing $N$, outperforming the single-shot attack. For certain sequences of measurement choices in the optimal strategy, such as those that have already given $N-l'$ wrong answers, the impossibility of certain measurement choices is perfectly predictable. Nevertheless, knowledge of the hidden variable $x$ makes knowledge of previous measurement choices irrelevant to the cheating strategy.

\begin{figure}[!tb]
\begin{center}
\includegraphics[width=3.4in,natwidth=5.71in,natheight=3.48in]{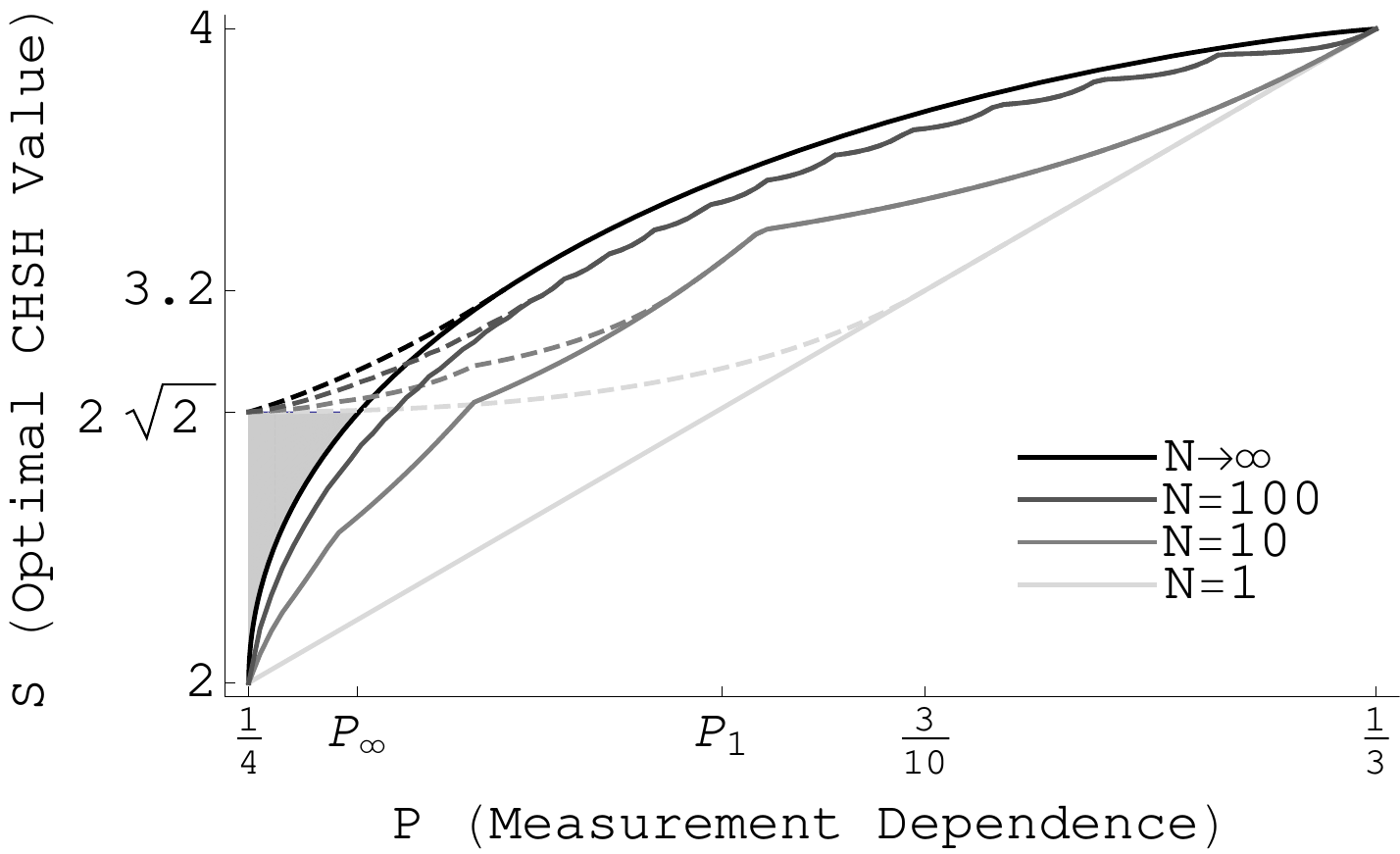}
\end{center}
\vspace{-0.5cm}
\caption{The maximal CHSH expectation value $S$ for given MD parameter $P$ with varying numbers of runs $N$ across which measurement basis choices are correlated. Strategies are classical (solid lines) or quantum (dashed lines). The black lines are the correlated strategies in the limit $N\rightarrow \infty$, therefore upper bounds for all finite $N$-run strategies. The shaded region indicates achievable CHSH violations using quantum technology and perfect measurement independence.} \label{fig:main}
\vspace{-0.5cm}
\end{figure}

The optimal strategy's behaviour in the large $N$ limit is less clear. Is the limiting curve simply the $S=4$ line, making perfect measurement independence the singular point at which the CHSH test functions? The following theorem answers this in the negative.

\begin{theorem}
The measurement dependence $P$ required to simulate a CHSH score $S$ with a deterministic strategy correlated over $N$ runs, in the $N\rightarrow \infty$ limit, has the lower bound
\be P\geq \left(\frac{4+S}{24}\right)^{\frac{4+S}{8}} \left(\frac{4-S}{8}\right)^{\frac{4-S}{8}} \label{PSbound}. \ee
\end{theorem}
\begin{proof} It is enough to consider the $N+1$ points defined by \eqref{Pinlterms} and \eqref{Sinlterms} and find lower/upper bounds for $P$/$S$ respectively as functions of the rescaled parameter $l=l'/N \in [3/4,1]$. To bound $P$, observe that $P= \frac{1}{4} \Pr(X \geq l')^{-1/N}$ where $X$ is a binomial distribution with $N$ trials and success probability $3/4$. The additive form of the Chernoff bound easily recovers 
\be \label{Plbound}
P \geq ( l/3 )^l (1-l)^{1-l}.
\ee
We aim to bound $S$ by dividing the region $k=l'\ldots N$ into two runs, split at $\alpha N$, such that
\begin{eqnarray*}
R_1 &:=& \sum_{k=l'}^{\alpha N-1}3^k\binom{N}{k} \geq 3^{lN}\binom{N}{lN} \geq 3^{lN}e^{NH(l)}, \\
R_2 &:=& \sum_{\alpha N}^N 3^k\binom{N}{k} \leq 4^Ne^{-ND(\alpha\|\frac{3}{4})},
\end{eqnarray*}
where we have used Stirling's approximation and a Chernoff bound respectively. $H(l)=-l \log_2 l -(1-l)\log_2(1-l)$ is the binary entropy while $$ D(\alpha\|\beta) = \alpha \ln \left(\frac{\alpha}{\beta}\right) + (1-\alpha)\ln \left(\frac{1-\alpha}{1-\beta}\right) $$ is the Kullback-Leibler divergence between Bernoulli distributed random variables with parameters $\alpha$ and $\beta$. For any $\alpha=l+\epsilon$ with finite (but small) $\epsilon$,
$$
\frac{R_2}{R_1}\leq\frac{4^Ne^{-ND(\alpha\|\frac{3}{4})}}{3^{lN}e^{NH(l)}}
$$
is exponentially small in $\epsilon N$. For increasing $\frac{R_2}{R_1}$,
$$
S\leq 8\frac{\alpha R_1+R_2}{R_1+R_2}-4
$$
is increasing. In the large $N$ limit this yields
\be \label{Slbound}
S\leq 8(l+\epsilon)-4.
\ee
The bound is tight since a lower bound for \eqref{Sinlterms} is given by \be S\geq 8\frac{P^Nl'}{N}\sum_{k=l'}^N 3^k \dbinom{N}{k} -4=8l-4, \ee substituting the definition of $P$ in \eqref{Pinlterms}. Substitution for $l$ in the bounds \eqref{Plbound} and \eqref{Slbound} yields the result.
\end{proof}
The bound may alternatively be expressed in terms of the CHSH winning probability $p_{\mathrm{win}}$ as \be P\geq \left(\frac{p_{\mathrm{win}}}{3}\right)^{p_{\mathrm{win}}} (1-p_{\mathrm{win}})^{1-p_{\mathrm{win}}}. \ee
The regime of allowable correlated LHV attacks described by Theorem 1 is depicted by the solid lines in Figure \ref{fig:main}. The points $P_{\infty}\approx 0.258$ and $P_1 \approx 0.285$ are where a CHSH value of $S=2\sqrt{2}$ can be achieved with (infinitely) correlated strategies and single shot attacks respectively. In the shaded region (for $P<P_{\infty}$), the $S$ values are below $2\sqrt{2}$ and therefore achievable with quantum technologies, and yet can never be simulated by any arbitrarily correlated attack without detection from Alice and Bob in their observed measurement choices.

\subsection{B. Quantum adversary} Tsirelson's bound proves that a test of the CHSH inequality with perfect measurement independence has a maximum possible value of $2\sqrt{2}$, for instance when the experimenters perform quantum measurements on a singlet state \cite{26}. Such quantum strategies are known to be vulnerable to a limited-free-will attack in the single shot case by optimising over both the measurement input distribution and the measurement operators themselves \cite{19}. Further correlating over $N$ runs of the experiment can also enhance this quantum strategy. 

%Whilst the outcomes of later runs can in principle be allowed to depend on the outcomes of previous runs through use of a quantum state entangled between all of the devices, here 
We make the common assumption that each run is causally disconnected \cite{27,28}, which ensures that the measurement choices for each CHSH test, given knowledge of the hidden variable, are otherwise independent, i.e.\ measurement choices at distant locations are not known beyond the extent to which it is implied by knowledge of the hidden variable. The causal independence also enforces that the effective CHSH operator being optimised has the local form \begin{equation*}\mathcal{S}_N:=\mathcal{S}\otimes \mathbb{1} \otimes \mathbb{1} \ldots+\mathbb{1} \otimes \mathcal{S}\otimes \mathbb{1} \ldots+ \mathbb{1} \otimes \mathbb{1} \otimes \mathcal{S}\ldots+\ldots .\end{equation*}Alice and Bob can only perform local measurements, thus use of an entangled state between the devices over multiple runs does not provide an advantage. The state that optimises this operator is separable, corresponding to a tensor product of the states that optimise $\av{\mathcal{S}}$ in \eqref{CHSHquantumscore}, such that $|\av{\mathcal{S}_N}|=N|\av{\mathcal{S}}|$, so it is enough to optimise over the single-shot operator $\mathcal{S}$. 

For a specified value of $P$, determining the optimal quantum strategy (which will typically correspond to no eavesdropping) will yield the maximum realisable CHSH value, which is important in bounding an eavesdropper's knowledge for a given CHSH value.

Using the notation $p_{\bs{y}}$ for the probabilities of a given set of $N$ measurement settings described by a string $\bs{y}\in \{0,1,2,3\}^N$, the aim is to maximise the expectation
\begin{equation}
\av{\mathcal{S}}=\frac{4}{N}\sum_{n=1}^N \left\langle p_0^n A_0B_0+p_1^n A_0B_1+p_2^n A_1B_0-p_3^n A_1B_1 \right\rangle \label{CHSHquantumscore}
\end{equation}
where, by our causality assumption, $p_m^n = \sum_{\bs{y}:y_n=m}p_{\bs{y}}$ are the marginal probabilities for a given setting $m$ being chosen on run $n$, and the measurement operators $A_j$ and $B_k$ are two-valued operators. Bounds on $\av{\mathcal{S}}$ subject to fixed $P$ have been derived for an individual run in \cite{19}. Using the same symmetry arguments as above, we reduce the set $\{p_{\bs{y}}\}$ to $\{p_k\}$ according to the number of correct answers each set $\bs{y}$ gives. This leaves $p_0^n=p_1^n=p_2^n=R$ and $p_3^n=1-3R$ where, through the symmetry-based reduction, $R$ is simply
$$
R=\sum_{k=1}^N3^{k-1}\frac{k}{N}\binom{N}{k}p_k
$$
which is equivalent to \eqref{Sinpkterms} up to constant factors. Therefore, for a fixed degree of MD requiring $p_k\leq P^N\;\;\forall k$, $R$ is optimised by the same input distribution (see Eq.\ (\ref{Pinlterms})). The achievable expectation values are then
$$
|\av{\mathcal{S}}|\leq\frac{4(1-2R)^{3/2}}{\sqrt{1-3R}} ,
$$
provided $R<3/10$ \cite{19}. This maximal value $S_Q$ can be compared with the value $S_C= 4(6R-1)$ obtained by the optimal classical, deterministic strategy as \be S_Q = \frac{2(8-S_C)^{3/2}}{3\sqrt{6(4-S_C)}} \ee when $S_C< 16/5$. For $R\geq 3/10$, hence $S_C\geq 16/5$, the two strategies coincide. The quantum strategies are compared with their classical equivalents in Figure \ref{fig:main}.

\section{V. Numerical Approach to Other Bell Inequalities} The analytical results for correlated strategies to the CHSH test above are not easily replicable for more complex Bell tests, since there are complications when a Bell test does not require the correlations produced by every possible measurement pair to evaluate the score. Nevertheless, as alluded to in the definitions, optimal finite $N$-run strategies can be determined numerically for a given measure of MD using linear programming. We briefly outline the approach for the measure $M_1$ (since $P$ is simpler), for both the CHSH test and its well-known generalisation to the class of $m$ setting, 2 outcome tests, $I_{mm22}$ \cite{29}.

For the CHSH test, the aim is to maximise $S$ as in \eqref{Sinpkterms} subject to normalisation \eqref{norminpkterms} and fixed MD $M_1$. By assuming $p(\bs{x})=4^{-N}$, the normalisation ensures $\sum_{\bs{x}} p(\bs{y}|\bs{x})=1 $ for each $\bs{y}$, whilst $M_1$ is given by \begin{equation} M_1 = \max_{\bs{x}} M_{\bs{x}} \qquad M_{\bs{x}}:=\sum_{\bs{y}} \left|p(\bs{y}|\bs{x})-\frac{1}{4^N}\right| \label{Mxydef}\end{equation}
We may again reduce the conditional probabilities to the $N+1$-element vector $\bs{p}=( p_k )$ by considering the following symmetry argument. Observe that to maximise $S$ for fixed $M_1$, all the $S_{\bs{x}}$ will be equal, which in turn means that all $M_{\bs{x}}$ will be equal. \eqref{Mxydef} reduces to \be M_1 = \sum_{k=0}^N \dbinom{N}{k}3^k \left|p_k -\frac{1}{4^N}\right| \label{MKpkdef}.\ee
Introducing new non-negative variables $\bs{w}$ such that \begin{equation} \label{wvariables} w_k \geq p_k - 4^{-N} \qquad , \qquad w_k \geq 4^{-N} -p_k \end{equation} allows removal of the modulus to create a standard linear programming problem by converting all inequalities to a statement on non-negative variables, i.e.\ introduce variables $\bs{a}$, $\bs{b}$ defined by $a_k :=w_k -p_k +4^{-N}$ and $b_k :=w_k +p_k -4^{-N}$ and set $\bs{a}\geq 0$, $\bs{b} \geq 0$ to represent $\eqref{wvariables}$. The measurement independence condition \eqref{MKpkdef} is then expressed as \begin{equation} \label{MKwkcondition} M_1 = \sum_{k=0}^N \binom{N}{k} 3^k w_k .\end{equation}

We have now reduced to the linear problem
\begin{eqnarray*}
\mathrm{minimize} & \bs{c}\cdot \bs{z} \\
\mathrm{subject \ to} & B\cdot \bs{z} = \bs{v} , \bs{z}\geq 0
\end{eqnarray*}
where $\bs{z}$ is a $4(N+1)$-element vector with the block structure  $\bs{z}^T=(\begin{array}{llll} \bs{p} & \bs{w} & \bs{a} & \bs{b} \end{array})$ and $\bs{c}^T=(\begin{array}{llll} -\bs{s} & 0 & 0 & 0 \end{array})$, where $s_k :=3^k \binom{N-1}{k-1}$, is used to maximise $S=-4-8 \bs{c}\cdot \bs{z}$. The normalisation constraint, the definitions of $\bs{a}$ and $\bs{b}$ and the measurement independence constraint \eqref{MKwkcondition} are found within
$$ B=\left( \begin{array}{cccc} 
\bs{n} & 0 & 0 & 0\\
\mathbb{1} & -\mathbb{1} & \mathbb{1} & 0\\
-\mathbb{1} & -\mathbb{1} & 0 & \mathbb{1} \\
0 & \bs{n} & 0 & 0\\
\end{array} \right) \qquad , \qquad \bs{v} = \left( \begin{array}{c} 1 \\ 4^{-N} \\ -4^{-N} \\ \frac{1}{N_1}M_K \\ \end{array} \right) $$
where $n_k:= \dbinom{N}{k}3^k$.

The optimal single-shot attack ($N=1$) under the measure $M_1$ is in fact the same model as before with parameter $P\in[1/4,1/3]$, where $M_1=3(P-1/4)$ and thus $S=2+8M_1$. We compare the optimal correlated $N$-run attack with $N$ repetitions of the one-shot attack. As with the analysis for $P$, both strategies coincide at the limits $S=2$ and $S=4$, the latter being achieved with MD of $M_{\max}(N):= 2(1-(3/4)^N)$, again with $p_N = 3^{-N}$ and otherwise $p_k =0$. The comparison for $N=100$ is shown in Figure \ref{CHSHsimplexplot}. Unfortunately, the complexity of the $M_1$ measure compared to $P$ makes the derivation of a general bound for all $N$ difficult.
\begin{figure}[!t]
\begin{center}
\includegraphics[width=3in,natwidth=4.88in,natheight=3.46in]{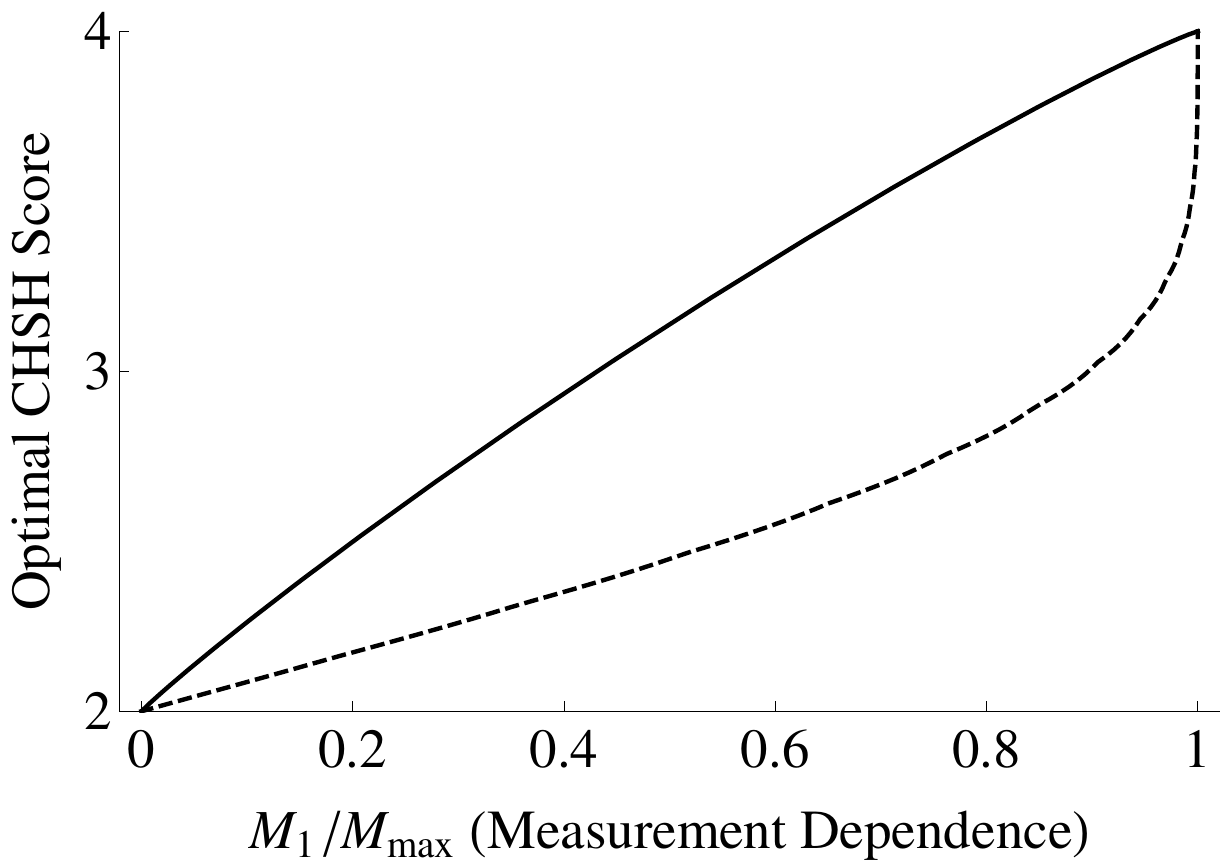}  
\end{center}
\vspace{-0.5cm}
\caption{Comparison of the optimal correlated attack (solid line) with $N$ repetitions of the optimal one-shot attack (dashed line) for $N=100$, using the CHSH inequality. The maximum score $S=4$ can be simulated classically with a measurement dependence of $M_1=M_{\max}:=2(1-(3/4)^N)$.\label{CHSHsimplexplot}}
\vspace{-0.5cm}
\end{figure}

In principle, any Bell test (linear function of the underlying probability distribution) can be expressed in this way. The family of $I_{mm22}$ Bell tests \cite{29}, which are two party, $m$ setting, two outcome ($\pm 1$) tests, benefit from a similar symmetry reduction to that of the CHSH test, corresponding to $m=2$. As an extension of \eqref{CHSHscore}, these tests can similarly be assigned a score \begin{equation} \label{Imm22score} S_{mm22}=m^2 \sum_{x}p(x)\sum_{j,k =0}^{m-1} p(A_j ,B_k |x )\alpha_{jk}^m a_j(x)b_k(x) \end{equation} where $\alpha_{jk}^m$ takes the value $+1$ if $j+k<m$, $-1$ if $j+k=m$ and $0$ if $j+k>m$. 

There are additional complications for $m\geq 3$, however. Firstly, not every pair of correlations is assigned the same weight; some pairs are not used at all, i.e.\ $(j,k)$ for which $\alpha_{jk}^m=0$. Nevertheless, an eavesdropper needs to ensure that such correlations still arise so as not to be suspicious. This means that, upon enumerating the measurement settings with strings $\bs{y}\in \{0,\ldots ,m^2-1\}^N$ and determining the optimal outcome sets $\bs{x}$, the symmetry structure of $p(\bs{y}|\bs{x})$ has two parameters -- $k$, the number of `correct' answers (as before) to correlation pairs that are used, and $l$, the number of unused correlation pairs in the string $\bs{y}$. Hence, we have a set of probabilities $\{p_{k,l}\}$ to optimise over. However, this also means that the constraints \eqref{Bayescondition} only reduce to a set of $N+1$ conditions dependent on how many unused correlation pairs are present for given settings $y$ rather than a single condition, and depending on the choice of distribution $p(\bs{x})$, it may never be possible to satisfy all these conditions. Nevertheless, this choice does not affect optimality; provided a distribution of $p(\bs{x})$ is selected such that the constraints \eqref{Bayescondition} are all fulfilled, that is equally as good as any other satisfying assignment, and the same optimal values can be achieved. Figure \ref{I3322simplexplot} shows a comparative plot for the $I_{3322}$ inequality in which an advantage to the correlated attack is found.
\begin{figure}[!t]
\begin{center}
\includegraphics[width=3in,natwidth=4.88in,natheight=3.46in]{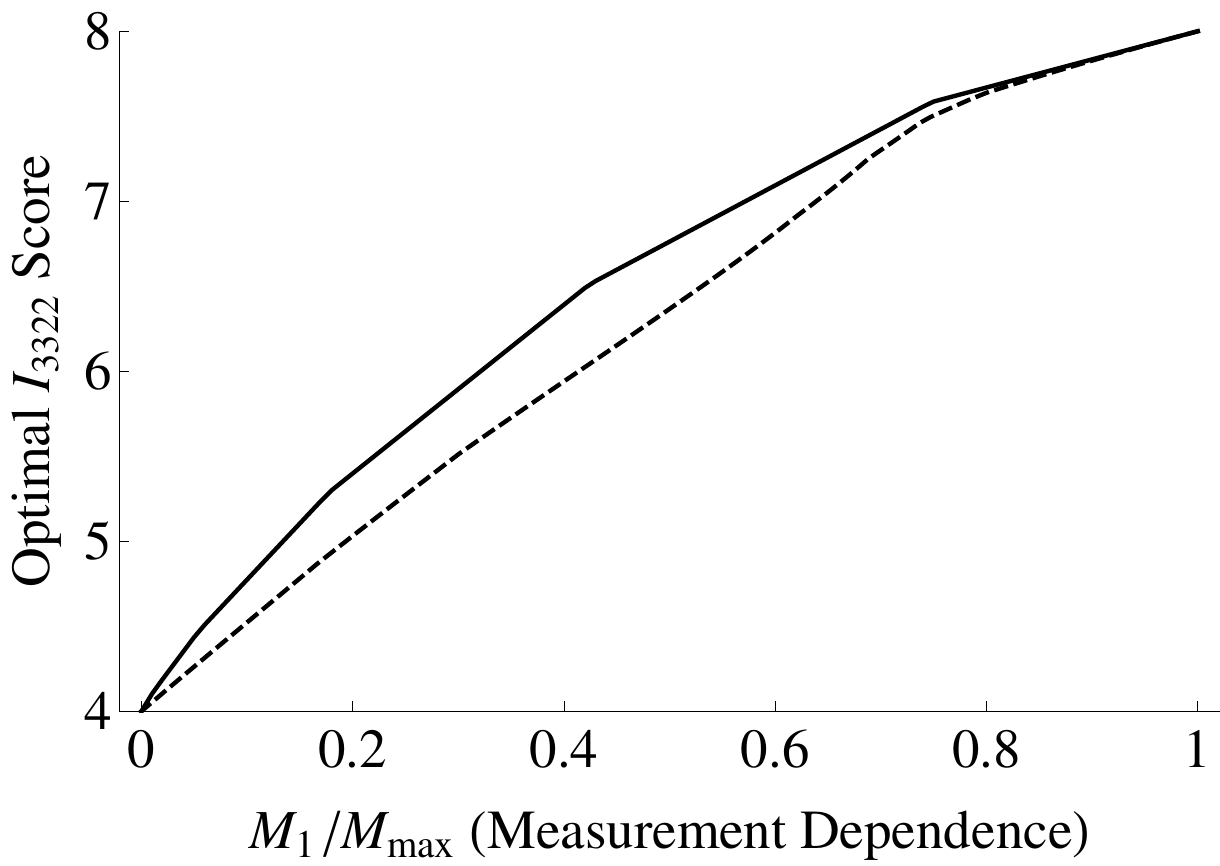}  
\end{center}
\vspace{-0.5cm}
\caption{Comparison of the optimal correlated attack (solid line) with $N$ repetitions of the optimal one-shot attack (dashed line) for $N=10$, using the $I_{3322}$ inequality. The maximum score $S_{3322}=8$ can be simulated classically with a measurement dependence of $M_1=M_{\max}:=2(1-(7/9)^N)$. \label{I3322simplexplot}  }
\vspace{-0.5cm}
\end{figure}

\section{VI. Conclusions}
By manipulating the RNGs used to select measurements in a Bell test in tandem with the devices performing them, an adversary may simulate a Bell violation. The degree of violation can be bounded above in terms of an appropriate measure of MD. Crucially, whilst the adversary gains a significant advantage in employing attacks correlated over many runs of an experiment, as opposed to single-shot attacks, there are still violations which cannot be reproduced by such attacks if the experimenters' degree of MD is sufficiently low. In light of this, existing analyses of the working regimes for device-independent randomness expansion \cite{19} and key distribution protocols could be revised, although the problem of performing privacy amplification without trusted randomness would need addressing. Since many of these utilise the CHSH test, our focus on this test is immediately applicable, while application to other tests is a simple linear programming problem. How to experimentally assess the value of $P$ (or another measure) in a pair of RNGs remains an open question.

While our analysis does not require the RNGs to be Santha-Vazirani sources \cite{20}, as in proposed randomness amplification protocols \cite{16,17,18}, our results can be interpreted in the context of randomness amplification. For a given MD $P$, the optimal quantum strategy, in the regime where quantum beats classical (i.e.\ $S_C<16/5$), gives perfectly random measurement outcomes on one side. If the players could know the value of the hidden variable in a run, and therefore the measurement selection bias, they can implement the optimal quantum strategy, which has the potential to allow perfect amplification (i.e.\ procuring perfectly random bits from partially random bits) in this regime. However, there is no obvious reason why honest players would have such knowledge of the hidden variables.

\section{Acknowledgements} JEP is supported by an EPSRC postgraduate studentship. We would like to thank Artur Ekert for useful discussions.


\begin{thebibliography}{99}
\bibitem{22} $p=2$ can be analysed using quadratic programming.
\end{thebibliography}
\end{document}